\documentclass{article}
\usepackage{tabularx}
\usepackage{booktabs}
\usepackage{datetime}
\usepackage{amsfonts,amsmath,algorithm,amsthm}
	\usepackage[noend]{algorithmic}

\newcommand{\approach}[1]{\textsf{Auction algorithm}}
\newcommand{\lsa}[1]{\textsf{LSA}}
\newcommand{\baseline}[1]{\textsf{Baseline}}

\newcommand{\comment}[1]{}

\newtheorem{proposition}{Proposition}

\newcommand{\ignore}[1]{}

\newtheorem{theorem}{Theorem}
\newtheorem{lemma}{Lemma}
\newtheorem{corollary}{Corollary}

\begin{document}
\title{Revisiting the Auction Algorithm for Weighted Bipartite Perfect Matchings}
%
%
\author{Megha Khosla and Avishek Anand}

%

%
\date{L3S Research Center, Hannover, Germany \\ {lastname}@l3s.de}
\maketitle   
\begin{abstract}
    We study the classical \emph{weighted perfect matchings problem} for bipartite graphs or sometimes referred to as the \emph{assignment problem}, i.e., given a weighted bipartite graph $G = (U\cup V,E)$ with weights $w : E \rightarrow \mathcal{R}$ we are interested to find the maximum matching in $G$ with the minimum/maximum weight. In this work we present a new and arguably \emph{simpler} analysis of one of the earliest techniques developed for solving the assignment problem, namely the \emph{auction algorithm}. Using our analysis technique we present tighter and improved bounds on the runtime complexity for finding an approximate minumum weight perfect matching in  $k$-left regular sparse bipartite graphs. 
\end{abstract}
\section{Introduction}
\label{sec:intro}
Let $G = (U\cup V,E)$ be a bipartite graph on $n$ vertices and $m$ edges and let $w:E\rightarrow \mathcal{R}$ be a weight function on edges. A perfect matching of $G$ is a subset $M\subseteq E$ of the edges such that for every node $v\in U\cup V$ there is exactly one incident edge $e\in M$. The weight of a matching $M$ is given by the sum of the weights of its edges, i.e. $w(M) := \sum_{e \in M} w(e)$. The minimum weight bipartite perfect matching problem (MWPM) is to find for a given bipartite graph $G$ and a given weight function $w$ a perfect matching of minimum weight which is sometimes also referred to as the \emph{assignment problem}. WLOG we can assume that $|U|\le |V|=n$.

In this work we present a novel and simple analysis of \emph{auction} algorithm originally proposed by Bertsekas~\cite{Ber81} for solving the assignment problem. The auction algorithm resembles a competitive bidding process whereby unassigned persons (nodes on the left set $U$) bid simultaneously for objects (nodes on the right set $V$), thereby raising their prices. On obtaining all bids, the object $v \in V$ is assigned to the highest bidder. 

There is a cost $w(u,v)$ for matching person $u$ with object $v$ and we want to assign persons to objects so as to minimize the overall cost (for a minimization objective). 
Let $E$ be the set of all pairs $(u,v)$ that can be matched.
A typical iteration of the auction algorithm consists of a \textbf{bidding} and an \textbf{assignment} phase. To start off, each object $v$ is initialized with some initial price $L(v)$. The bidding and assignment phases are as follows:

\textbf{Bidding Phase :} Let $I$ be the set of unassigned persons. Each person $u \in I$ finds an object $\bar{v}_u$ which optimizes for minimum cost, $w(u,v) + L(v)$ , that is,
$$ \bar{v}_u \,=\, \arg \min_{v:(u,v) \in E} \,w(u,v) + L(v)$$ 
and computes a \emph{bidding increment} $\gamma_u$ for some parameter $\varepsilon>0$
$$\gamma_u = w^2_u - w^1_u +\varepsilon,$$
where 
$$w^1_u = \min_{v: (u,v)\,\in\, E} w(u,v) + L(v)$$ 
is the best object value and 
$$w^2_u = \min_{v: (u,v)\,\in \,E, \,v \neq \bar{v}_u} w(u,v) + L(v)$$ 
is the second best object value.

\textbf{Assignment Phase :} Note that an object $v$ can be the best bid for multiple persons $u \in I$. In such a case $v$ is assigned to the person with the \emph{highest bid} and its price is raised by the highest bid, i.e., $$L(v) = L(v)+\max_{u \in I : v = \bar{v}_u} \gamma_u.$$ 
The person that was assigned to $v$ at the beginning of the iteration (if any) becomes unassigned. The algorithm continues with a sequence of iterations until all persons have an assigned object.

 The auction algorithm finds an approximate solution for maximum/minimum weight perfect matching which is bounded by $OPT + n\varepsilon$, which we also refer as the \emph{$\varepsilon$- optimal} solution. 
Let $w^{max}$ and $w^{min}$ denote the maximum and minimum edge weights respectively.
The worst case running time for complete bipartite graphs is $O({n^2 w^{max} \over \varepsilon})$.  In order to improve the running time a procedure called \emph{$\varepsilon$-scaling} is employed. 
It essentially consists of executing the algorithm several times starting with a large value of $\varepsilon$ and successively reducing $\varepsilon$ after each run. The procedure terminates typically when the ultimate value of $\varepsilon$ is less than some critical value (for example, $1/n$, for integral weights).
 For integral weights, it has be shown that the worst-case running time of the auction algorithm for finding optimal solutions using $\varepsilon$-scaling is $O(nm\log(nw^{max}))$~\cite{BeE88,BeT89}. However for the asymmetric problem where number of persons is less than the number of objects, the prices would need to be initialized by $0$ and $\varepsilon$-scaling cannot be used out of the box. 

In order to benefit from  $\varepsilon$-scaling an \emph{reverse} auction method is used (in addition to the regular auction) in which the objects also compete for persons by offering discounts~\cite{BCT91}. 

Subsequently the auction method was extended to solve the classical linear network flow problem and many of its special classes. In particular, \cite{Ber86a} and \cite{Ber86b} propose an extension to the minimum cost network flow problem using $\varepsilon$-relaxation. Auction algorithms for transportation problems~\cite{BeC89a} and shortest paths~\cite{Ber91} have also been proposed. We refer the interested reader to~\cite{Bertsekas1992} for a more comprehensive discussion on auction algorithms.

In this paper we focus on the sequential version of auction algorithm where a single unassigned person bids at every iteration. This version is also known as the \emph{Gauss-Seidel} version because of its similarity with Gauss-Seidel methods for solving systems of nonlinear equations. Moreover, we restrict ourselves to its application to finding minimum weight perfect matchings in bipartite graphs. 
Our analysis technique is inspired by the \lsa{} method~\cite{khosla2013balls} which is used to construct large hash tables and finding maximum matchings in unweighted bipartite graphs. In fact, the original motivation was to extend the label based technique in \lsa{} for solving the weighted version of the matching problem in bipartite graphs. Though our proposed algorithm later turned out to be a version of the auction algorithm, our analysis technique allows for a simpler interpretation of prices (labels) as a function of \emph{shortest unweighted paths} in the underlying bipartite graphs. This in turn helps us to bound the runtime of the weighted version in terms of the shortest unweighted paths in the underlying graphs (which also provides a bound on the runtime of LSA) with weight range ($w^{max}-w^{min}$) in the multiplicative factor.

In particular, we use the main result in \cite{khosla2013balls} to show that the worst case runtime of auction algorithm for sparse $k$- regular bipartite graphs is no more than $O(n\cdot {w^{max}-w^{min}\over \varepsilon})$ with high probability. For complete bipartite graphs, we prove a slightly better runtime bound of $O(n^2\cdot {w^{max}-w^{min}\over \varepsilon})$, which is an improvement over the previous bound when both $w^{max}$ and $w^{min}$ are large.

\subsection{More on Related Work}
\label{sec:rel-work}

The first polynomial time algorithm for the assignment problem, the so called Hungarian method, was given by ~\cite{Kuhn55,Kuhn56}; implementations of this algorithm had a running time of $O(n^3)$, which is still optimal for dense graphs. 
Edmonds and Karp~\cite{EdmondsK72} and Tomizawa~\cite{Tomizawa71} independently observed that the assignment problem is reducible to computing single-source shortest paths on non-negative weights. Further, Fredman and Tarjan~\cite{FT87} showed that using Fibonacci heaps, $n$ executions of Dijkstra's~\cite{Dij59} shortest path algorithm take $O(mn+n^2\log n)$ time. On integer weighted graphs this algorithm can be implemented slightly faster, in $O(mn + n^2\log\log n)$ time~\cite{Han02,Tho03} or $O(mn)$ time (randomized)~\cite{AnderssonHNR98,Thorup07b}, independent of the maximum edge weight. Gabow and Tarjan~\cite{GT89} gave a scaling algorithm for the assignment problem running in $O(m\sqrt{n}\log(nw^{max}))$ time for integral weights that is based on the Hungarian method. Orlin and Ahuja~\cite{OrlinA92} using the {\em auction} approach, and Goldberg and Kennedy~\cite{GoldbergK97}, obtain the same time bound. Recently the new scaling algorithm by Duan et al.~\cite{Duan:2017} matches the same bound for weighted matchings in general graphs. Sankowski takes an algebraic approach by proposing a randomized algorithm~\cite{Sankowski09}, that solves the assignment problem using fast matrix multiplications in $O(w^{max}n^\omega)$ time with high probability, where $\omega$ is the exponent of square matrix multiplication. Improved runtime bounds have also been achieved by reducing maximum weighted matching to maximum cardinality matching problems ( see~\cite{kao2001,huang2012,pettie2012} and references therein)

The other area of related work concerns with the algorithms for maximum weighted matchings (MWM). It is important to note that though the two problems MWM and MWPM are reducible to each other, the reductions do not work for their approximate versions as the approximation may compromise perfection. Duan and Pettie~\cite{duan2014linear} presented a $(1-\varepsilon)$-MWM algorithm that runs in $O(m\varepsilon^{-1}\log(\varepsilon^{-1}))$ on general graphs and $O(m\varepsilon^{-1} \cdot \min\{\log \varepsilon^{-1},\log (w^{max})\})$ time on integer-weighted general graphs. These results cannot be compared with the bound that we provide for random sparse bipartite graphs as we require that the matching returned should also be perfect. Please refer to~\cite{duan2014linear} for more details on other works for finding approximate maximum weight matching.

\subsection{Our Contribution}
We present a novel analysis of auction algorithm for finding approximate minimum weight perfect matchings in bipartite graphs for arbitrary weights. Our approach is inspired by the \emph{local search allocation} (\lsa{}) method~\cite{khosla2013balls} used to construct large hash tables and finding maximum matchings in unweighted bipartite graphs. Using our analysis we could easily provide a better runtime bound for a random sparse $k-$ left regular bipartite graphs. In a random $k$-left regular bipartite graph, each vertex in the left-set chooses $k$ neighbors from the right set independently and uniform at random.

From~\cite{inp:l12,fp12,fm12}  we know that for such graphs, for all $k\ge 3$, there exists a threshold density $c^*_k$ ( and is computable , for example for $k=3$ it is close $0.91$) such that when $|U|/|V| < c^*_k$, there exists a left perfect matching in $G$ with high probability, otherwise this is not the case. We will show that for random $k$- left regular bipartite graphs obeying this threshold condition, one can find an $\varepsilon$- optimal minimum weight perfect matching in near linear time. In case a perfect matching does not exist, we will show that with appropriate stopping criteria, the algorithm will stop and return an approximate solution for minimum weight maximum matching (though we do not analyse the approximation guarantee in detail). We note that there exists no restriction on the weights and they can be arbitrary. For complete bipartite graphs the runtime bound is $O(n^2 \cdot {w^{max}-w^{min} \over \varepsilon})$.
We now state the main result. 

\begin{theorem} \label{thm:main}
Let $G=(U\cup V, E)$ be a weighted bipartite graph with $|U|\le |V| = n$ and arbitrary weights. Let $OPT_G$ be the weight of the optimal minimum weight maximum matching in $G$. Then for an arbitrary parameter $0<\varepsilon $ the auction algorithm returns a $OPT_G + n\varepsilon$ minimum weight maximum matching. The worst case runtime for sparse and complete bipartite graphs is as follows.
\begin{enumerate}

\item \textbf{Sparse Graphs:} For $k\ge 3$, let $G=(U\cup V,E)$ be such that each vertex in $U$ chooses $k$ neighbors in $V$ independently and uniform at random. In addition let $|U| < c^*_k n$. The worst case runtime in case a perfect matching exists is $O(n \cdot 
{w^{max}-w^{min} \over \varepsilon})$ with probability $1-o(1)$. 

\item \textbf{Complete Graphs:} For complete bipartite graphs the runtime bound is $O(n^2 \cdot {w^{max}-w^{min} \over \varepsilon})$.

\end{enumerate}
\end{theorem}

The main idea behind the proof for the runtime bounds is to show that the prices (which we will refer to as labels in the subsequent section) when initialized with zero are increasing by at least $\varepsilon$ and are bounded in terms of shortest distances in the underlying graph and the weight range. We will then use the result from~\cite{khosla2013balls} which provides a linear bound in expectation and with high probability for the sum of the shortest distances in the special class of sparse bipartite graphs as considered in Theorem~\ref{thm:main}. For the approximation guarantee, we will show that no alternating path allows for a decrease in the existing matching weight by more than $n\varepsilon$. Note that we do not provide any improved bound on approximation quality of the solution in this paper. We start by giving a detailed description of the algorithm in the next section.

\section{The Auction Algorithm for Perfect Weighted Matchings}
\label{sec:approach}

\subsection{Notations.} Throughout the paper we use the following notations unless stated otherwise. We denote the set of integers $\{1,2,\ldots, n\}$ by $[n]$. Let $G =(U\cup V;E)$ denote a weighted bipartite graph where $|U|\le |V|=n$.
For any $e\in E$, $w(e)$ denotes the weight on edge $e$. In addition let $w^{max}$ and $w^{min}$ denote the maximum and minimum weight on any edge respectively. For any vertex $u\in U$ we refer to set of neighbors of $u$ in $V$ as $N(u)$. For any $v\in V$, $L(v)$ denotes the label of $v$. Let $M$ be the set of matched edges in the optimal matching. For $u\in U$ and $v\in V$, we say that an edge $e= (u,v)$ is assigned to $v$ if $e\in M$. We call a vertex $v\in V$ \emph{free} if and only if no edge is assigned to it.
The optimal minimum weight for a perfect matching in $G$ is then equal to $\sum_{e\in M} w(e)$ and is denoted by $OPT_G$.

\subsection{The Algorithm and its Analysis}\label{sec:algo}

We first explain the auction algorithm in detail with respect to our interpretation. 
Initially we are provided with all vertices from the right set $V$. 
We assign labels to all vertices in $V$ and initialize them by $0$. The vertices of the left set $U$ appear one by one together with the incident edges and their weights. An edge $e=(u,v)$ is assigned to vertex $v\in V$ such that the sum $L(v) + w(e)$ is minimum for all edges incident on $u$. We refer to the rule for choosing a candidate edge as the \emph{choice rule}. The original algorithm aimed only to find prefect matchings. In case a perfect matching does not exist, the algorithm might enter in an endless loop. We will argue in the next section that the label values are bounded and therefore resolve this issue by checking the minimum label with the maximum possible value of $n/2 (w^{max}-w^{min} +\varepsilon)$. In case for some vertex $u$ all its neighbors have labels greater than the maximum value, $u$ is discarded and never matched.

Let for vertex $v'\in V\backslash \{v\}$ and $e'=(u,v')$, the sum $L(v') + w(e')$ is the minimum. For some $0<\varepsilon$, the label of $v$ after its new assignment will be updated as follows.
$$L(v) = L(v') + w(e') -w(e) + \varepsilon .$$
We call the above rule as the \emph{update rule}.

In case $v$ is not empty and was already assigned another edge $(u,v)$, the edge $(u,v)$ is moved out of the matching and the process is repeated for vertex $u$. 

Let $\mathbf{L}= \{ L(v_1), \ldots, L(v_n)\}$  and $\mathbf{T}= \{ T(v_1), \ldots, T(v_n)\}$ where $L(v_i)$ denotes the label of vertex $v_i$ and $T(v_i)$ denotes the vertex matched to $v_i$.  We initialize $\mathbf{L}$ with all $0$s , i.e., all vertices are free. \approach{} is described in Algorithm~\ref{algo} which calls the Procedure~\ref{algo:orientEdge} to match an arbitrary vertex from the left set  when it appears.

\begin{algorithm}[h!]
\caption{\approach{} ($U,V,E$)}
\label{algo}
\begin{algorithmic}[1]
\FORALL {$v\in V$}
\STATE { Set $L(v)=0$}
\STATE {Set $T(v)= \emptyset$}
\ENDFOR
\FORALL {$u\in U$}
\STATE CALL MatchVertex ($u, \mathbf{L},\mathbf{T}$)
\ENDFOR

\end{algorithmic}
\end{algorithm}
 \floatname{algorithm}{Procedure}
\begin{algorithm}[h!]
\caption{MatchVertex ($u, \mathbf{L},\mathbf{T}$)}
\label{algo:orientEdge}
\begin{algorithmic}[1]
\STATE Choose $v\in N(u)$ such that $L(v) + w(u,v)$ is the minimum ~~~~~~$\rhd${\textbf{Choice Rule}}
\IF {$L(v) > n/2 (w^{max} -w^{min} +\varepsilon)$}
\STATE RETURN
\ENDIF
\STATE $L(v) \leftarrow  \min{(L(v') + w(u,v')| v' \in N(u)\setminus \{v\})} - w(u,v)+ \varepsilon$~ ~~~~~~~~~~~$\rhd${\textbf{Update Rule}}
\IF{$(T(v)\neq \emptyset )$}
\STATE $y\leftarrow T(v)$~~~~~~~~~~~~~~~~~~ $\rhd${\textbf{Move that moves an edge out of matching}}
\STATE $T(v) \leftarrow u$  ~~~~~~~~~ $\rhd${\textbf{Move that assigns a new edge or matches a new vertex}}
\STATE $\mathbf{CALL}$ {MatchVertex($y, \mathbf{L},\mathbf{T}$)}
\ELSE  
\STATE $T(v) \leftarrow u$ ~~~~~~~~~~~~~~~~~~ $\rhd${\textbf{Move that assigns an edge or matches a vertex}}
\ENDIF
\end{algorithmic}
\end{algorithm}
We observe that if \approach{} does not enter an endless loop in the Procedure~\ref{algo:orientEdge}, it would return a matched vertex for each of the vertices in $U$. This implies that \approach{} will find a left perfect matching if it terminates.

In the next section we will prove that the algorithm does not run in endless loops, i.e., it terminates, by showing that the (1) labels are non decreasing and (2) labels are bounded by a maximum value. 
\subsection{Bounding the Labels} \label{sec:proof}
We need some additional notation. In what follows  a \emph{move} denotes either assigning an edge to a free vertex or replacing a previously assigned edge. Let $P$ be the total number of moves performed by the algorithm. For $p\in [P]$ we use $L_p(v)$ to denote the label of vertex $v$ at the end of the $p$th move. Let $\mathcal{M}_p$ denote the set of matched edges at the end of the $p$th move.

Let an edge $(u,v)$ is assigned to a vertex $v\in V$ in some move $p$. The choice and the update rules can then be rewritten as follows. 

\begin{align} \label{eq:chrule}
&\textbf{Choice Rule : } L_{p-1}(v) + w(u,v) \le \min_{v'\in N(u)\setminus \{v\}} (L_{p-1}(v') + w(u,v')) .
\end{align}
\begin{align} \label{eq:uprule}
&\textbf{Update Rule : } L_{p}(v)  = \min_{v'\in N(u)\setminus \{v\}} (L_{p-1}(v') + w(u,v')) -w(u,v) +\varepsilon.
\end{align}
We will first show that the label of exactly one vertex increases by at least $\varepsilon$ at the end of a move. The labels of other vertices remain unchanged. Even though this fact is obvious from the algorithm description, we prove it here for completeness.

\begin{proposition}\label{prop:lev}
  For any $ p \in [P]$ there exists exactly one vertex $v\in V$ such that  $L_p(v) \ge L_{p-1}(v) +\varepsilon$ and for all other $v'\in V\backslash \{v\}$, $L_p(v') = L_{p-1}(v')$.
\end{proposition}
\begin{proof}
By definition of the move, the label of exactly one vertex is altered at the end of a move. Let $v$ be assigned as edge $(u,v)$ is some move $p\in [P]$. Therefore labels of all vertices except $v$ remain unchanged at the end of the $p$th move, i.e., 
$$\forall v'\in V\backslash \{v\} : L_p(v') = L_{p-1}(v')$$
For vertex $v$ the new label at the end of the $p$th move is defined by \eqref{eq:uprule} (the update rule) as
$$ L_p(v) = \min_{v'\in N(u)\backslash \{v\}}( L_{p-1}(v') + w(u,v') )-w(u,v) +\varepsilon,$$ 
which combined with \eqref{eq:chrule} gives $L_p(v) \ge L_{p-1}(v) +\varepsilon$, thereby concluding the proof.

\end{proof}
In the following proposition we bound the label of any vertex with respect to the labels of the other neighbors of its matched vertex and the corresponding edge weights. We will need this to bound the maximum label of any vertex at the end of the $(P-1)$th move.
\begin{proposition}\label{prop:labweight}
 For all $ p \in [P]$ and all $(u,v)\in \mathcal{M}_p$, the following holds.
 $$L_p(v) \le \min_{v'\in N(u)\setminus \{v\}}(L_{p}(v') + w(u,v')) - w(u,v)+\varepsilon.$$
\end{proposition}
\begin{proof}
Let in some move $p\in P$ an edge $(u,v)$ is placed in the matched set $\mathcal{M}_p$. Note that the labels of all vertices $v' \in V\setminus \{v\}$ remain unchanged at the end of the $p$th move. By update rule we obtain

\begin{align*}
L_{p+1}(v) = &\min_{v'\in N(u)\backslash \{v\}}(L_p(v') + w(u,v')) - w(u,v) +\varepsilon \\
 \le & \min_{v'\in N(u)\backslash \{v\}} (L_{p+1}(v') + w(u,v')) - w(u,v) + \varepsilon.
\end{align*}
The last inequality holds $L_{p+1}(v') \ge L_{p}(v') $ for all $v'$ and $p$ by Proposition~\ref{prop:lev}.
For any other edge $(u',v')\in \mathcal{M}_p$ that was last assigned in some move $p'<p$ the following holds by update rule
 \begin{align*} L_{p'}(v') =  \min_{v''\in N(u')\backslash \{v'\}}(L_{p'-1}(v'') + w(u',v'')) - w(u',v') +\varepsilon \\
 \le \min_{v''\in N(u')\backslash \{v'\}}(L_{p}(v'') + w(u',v'')) - w(u',v') +\varepsilon
 \end{align*}
The last inequality holds as $L_{p}(v'')\ge L_{p'-1}(v'') $ for all $v''$ and all $p> p'-1$. 
Also as $v'$ was not updated after $p'$th move, $L_{p'} = L_{p}$. We therefore obtain 
$$L_{p}(v') \le  \min_{v''\in N(u')\backslash \{v'\}}(L_{p}(v'') + w(u',v'')) - w(u',v') +\varepsilon ,$$
thereby completing the proof.

\end{proof}
\subsubsection{ Maximum Label and the Runtime}
We want to prove a bound of the labels of the vertices in terms of their shortest distances to the set of free vertices at the end of the $(P-1)$th move. 
We start by considering  an ordered alternating path of unmatched and matched edges. Using Proposition~\ref{prop:labweight} we will exploit the relationship between (a)label of the last vertex (b) label of the first vertex and (c) edge weights in this ordered path.

For an even $4\le t\le 2n-4$, let $B_t=(v_1, u_2, v_2,\cdots, u_{{t\over 2} +1}, v_{{t\over2}+1}) $ denote an alternating path of unmatched and matched edges at the end of some move $p$.
We assume that the first edge in the path is an unmatched edge. Let $\mathcal{M}_p(B_t)$ and $\mathcal{M}'_p(B_t)$ denote the set of matched and unmatched edges in $B_t$ at the end of some $p$th move. 
\begin{lemma}\label{lem:step1}
For all $t\le 2n-4$ and all paths $B_t$ as defined above, the following holds.
$$ L_p(v_{t\over2}) \le L_p(v_1) + \sum_{e\in\mathcal{M}'_p(B_t)} w(e) - \sum_{e\in\mathcal{M}_p(B_t)}w(e) + {t\varepsilon \over 2}$$
\end{lemma}
\begin{proof}
We will prove the lemma by induction on the path length $t$. For the case $t=2$ we have 
$B_2 = (v_1,u_2,v_2)$, where $(v_1,u_2)\in \mathcal{M}'_p(B_2)$ and $(u_2,v_2)\in \mathcal{M}_p(B_2)$. By Proposition~\ref{prop:labweight} we obtain
$$ L_p(v_2) \le L_p(v_1) + w(v_1,u_2) - w(u_2,v_2)  + \varepsilon.$$
Clearly the lemma follows for path length $2$. Now assume that the lemma is true for path length $t=2t'$, i.e., 
$$ L_p(v_{t'}) \le L_p(v_1) + \sum_{e'\in \mathcal{M}'_P(B_{2t'})} w(e') - \sum_{e\in \mathcal{M}_p(B_{2t'})} w(e) + {t'\varepsilon}$$

We will now prove the lemma for paths of length $2t'+2$. Consider the last vertex $v_{{t'}+1}$ in $B_{2t'+2}$. As it is matched to $u_{t'+1}$, by Proposition~\ref{prop:labweight} we obtain
$$L_p(v_{t'+1}) \le L_p(v_{t'}) + w(u_{t'+1}, v_{t'} ) - w(u_{t'+1}, v_{t'+1} )+\varepsilon.$$
Combining the above inequality with the induction hypothesis we obtain 
$$L_p(v_{t'+1}) \le L_p(v_1) + \sum_{e'\in \mathcal{M}'_P(B_{2t'+2})} w(e') -   \sum_{e\in \mathcal{M}_p(B_{2t'+2})} w(e) + (t'+1) \varepsilon,$$
hence completing the proof.
\end{proof}
We obtain the following corollary about the labels at the end of $(P-1)th$ move.
\begin{corollary} \label{cor:dist}
Let $d_{P-1}(v)$ denote the length of the shortest unweighted path to any free vertex in $G$ after move $P-1$ is completed. Then for all $v\in V$, $L_{P-1}(v) \le {d_{P-1}(v)\over 2} (w^{max} - w^{min} +\varepsilon)$.
\end{corollary}
\begin{proof}
We assume that the given bipartite graph is connected, otherwise we run the algorithm on connected components. We note that at the end of $P-1$th move, there is at least one free vertex in $V$. Note that the label of any free vertex is $0$. For any $v\in V$we use Lemma~\ref{lem:step1} considering its shortest alternating path to some free vertex and bounding the weights of unmatched and matched edges on this path by $w^{max}$ and $w^{min}$ respectively, we obtain the desired result, i.e., 
$$ L_{P-1}(v) \le 0 + {d_{P-1}(v) \over 2} ( w^{max} -w^{min} +\varepsilon)$$
\end{proof}

Before we prove the runtime bounds we will describe the main result from~\cite{khosla2013balls} which we use to bound the distance values $d_{P-1}(v)$. Khosla~\cite{khosla2013balls} considers the problem of assigning $m$ items to $n$ locations, such that each of the $m$ items chooses $k$ locations independently and uniformly at random. Each item needs to be assigned to one of its $k$ choices. It is easy to see that such an assignment instance represents a $k$-left regular bipartite graph and a valid assignment corresponds to a left-perfect matching (where all vertices from the left set have been matched). 

A label based approach, LSA is used to find such a perfect matching for the case $m < c^*_k n$, where $c^*_k$ is the \emph{threshold density} (known from \cite{inp:l12,fp12,fm12} before). The fact that $m < c^*_k n$ ensures that a left perfect matching exists with probability $1-o(1)$. 
The runtime of the algorithm is bounded by the sum of labels at the end of the algorithm which are in turn bounded by the shortest distances to the set of free vertices. Note that $c^*_k <1$ and there always exist at least one free vertex at the end of the algorithm. The distances are in turn bounded using some structural properties of the corresponding graphs. She shows that the sum of shortest distances to the set of free vertices is bounded by $O(n)$ with probability $1-o(1)$ . In a subsequent extension Anand and Khosla~\cite{DBLP:KhoslaA16} show that the result also holds in expectation. We state here their main result adjusted to the terminology and its requirement in this paper.

\begin{theorem}\label{thm:LSA}
 For $k\ge 3$, let $G=(U\cup V,E)$ be such that each vertex in 
 $U$ chooses $k$ neighbors in $V$ independently and uniformly at random. In addition let $|U| < c^*_k |V|$. Then for some
 $\delta>0$, $\sum_{v\in V} d_{P-1} (v) = O(n)$ with probability $1-n^{-\delta}$ and in expectation.
\end{theorem}
We refer the reader to \cite{khosla2013balls,DBLP:KhoslaA16} for a complete proof of Theorem~\ref{thm:LSA}. We are now ready to prove the runtime bounds of the auction algorithm.

\begin{proof}[Proof of runtime bounds (Theorem~\ref{thm:main})]
We recall that $P \ge \sum_{v\in V} L_{P-1}(v) +1.$
From Corollary~\ref{cor:dist} and Theorem~\ref{thm:LSA} we conclude that for sparse random $k$-left regular bipartite graphs $P\le O(n\cdot (w^{max}-w^{min} +\varepsilon))$ with high probability. We know that labels increase by at least $\varepsilon$ in each step. Further  $k$ comparisons are required, in each move, to find the best and the second best vertices. From these two observations we can conclude that the worst case bound is $O(n k\cdot {w^{max}-w^{min} +\varepsilon \over \varepsilon})$ with high probability.

For complete bipartite graphs, note that $d_{P-1}(v) \le 2$ for all $v \in V$ (for the last free vertex, it is $0$ )
Note that in each move $n$ comparisons are made which gives us the worst case bound of $O\left(n^2\cdot {w^{max}-w^{min} \over \varepsilon}\right)$.
\end{proof}
It is easy to see that for arbitrary bipartite graphs and also for the case where the perfect matching does not exist, the worst case runtime is bounded by $O(n^2d \cdot {w^{max}-w^{min} +\varepsilon \over \varepsilon})$, where $d$ is the maximum degree of the vertices in $U$. 
We note that in practice the maximum label value will be much lower than what we estimate here as we assume all unmatched edges to be with the largest weight and all matched edges to be of smallest weight. Also, what appears like the worst case for the runtime analysis, i.e., where each unmatched edge has weight $w^{max}$ and each of the matched edge has weight $w^{min}$ and when the weight ranges are large, is in fact a considerably easy case for the algorithm. 
The parameter $\varepsilon$ will not have much role to play in this case as for each label update since the label will be increased by a high value $w^{max}-w^{min}$.

For completeness, we prove in the following section that the algorithm outputs a $\varepsilon$-optimal solution for the case where the perfect matching exists. 
We do not go into the detailed analysis for the case where a perfect matching does not exist because of the, previously stated, bad worst case runtime bound.
We believe that a closer analysis will lead to the same approximation guarantee for this case too. In the future, we hope to improve the runtime bound by using scaling type approach, in which the maximum possible value of the label is initially set to a very small value. This would lead to discarding most of the vertices and would result in a small matching size. In each scale we would increase the maximum value by some factor and improve on the matching size obtained from the previous scale. 

\subsection{An $\varepsilon$-optimal solution}
Using Lemma~\ref{lem:step1} it is now to easy to show that \approach{} outputs a near optimal minimum weight perfect matching for all classes of bipartite graphs and arbitrary weights provided a perfect matching exists. 
\begin{lemma}\label{lem:outW}
For any $\varepsilon >0$ and given bipartite graph $G$, \approach{} outputs a perfect matching with weight at most $OPT_G + n\varepsilon$
\end{lemma}
\begin{proof}
Let us assume that \approach{} does not output an optimal answer. First consider the case where $|U|<|V|=n$, such that there is a free vertex $v\in V$ at the end of the last move. As in Lemma~\ref{lem:step1} let $B_{t} = (v,u_1,v_1,u_2,v_2,\cdots, u_t,v_t)$ be an augmenting path of length $2t$ where $\mathcal{M}_P(B_{t})$ and $\mathcal{M}'_P(B_{t})$ denote the set of matched and unmatched edges in $B_{t}$. By Lemma~\ref{lem:step1} we obtain
\[ L_P(v_t) \le L_P(v) + \sum_{e\in \mathcal{M}'_P(B_{t})}w(e)- 
 \sum_{e\in \mathcal{M}_P(B_{t})}w(e) + t\varepsilon.\]
 As $L_P(v)=0$ and  $L_P(v_t) > 0$ we obtain 
 \begin{equation}\label{eq:augPath}
 \sum_{e\in \mathcal{M}_P(B_t)}w(e) < 
 \sum_{e\in \mathcal{M}'_P(B_t)}w(e) + t\varepsilon \end{equation}

We next consider an augmenting cycle $C_{2t}$ of length $2t$ such that 
$$C_{2t} = (u_1,v_1, u_2,v_2,\cdots , u_t,v_t,u_1)$$ where for all $i\in [t]$, $(u_i,v_i)\in \mathcal{M}_P$. 
WLOG we assume that vertex $v_t$ was matched later than all other right set vertices in $C_{2t}$ in some move, say $p$.
Let $B_{2t-4} =(v_1,u_2,v_2,\cdots, u_{t-1},v_{t-1})$. Now by Lemma~\ref{lem:step1} we obtain
$$ L_{P}(v_{t-1}) \le L_{P}(v_1) + \sum_{e\in \mathcal{M}_P(B_{2t-4})} w(e) - \sum_{e\in \mathcal{M}'_P(B_{2t-4})} w(e) +(t-2)\varepsilon$$
As $v_{1}$ and $v_{t-1}$ by assumption were not updated after the $(p-1)$th move we can write the above inequality as 
$$ L_{p-1}(v_{t-1}) \le L_{p-1}(v_1) + \sum_{e\in \mathcal{M}_P(B_{2t-4})} w(e) - \sum_{e\in \mathcal{M}'_P(B_{2t-4})} w(e) +(t-2)\varepsilon,$$
because $L_P(v_1) = L_{p-1}(v_1)$  and $L_{P}(v_{t-1})=L_{p-1}(v_{t-1}) .$

By the choice rule we have 
\begin{align}\label{eq:augCycle}& L_{p-1}(v_t) + w(u_t,v_t) \le   L_{p-1}(v_{t-1}) + w(u_t,v_{t-1})  \nonumber \\
&\stackrel{}\le L_{p-1}(v_1) + \sum_{i=1}^{t-1} w(u_{i+1},v_i) - \sum_{i=2}^{t-1} w(u_i,v_i) + (t-2)\varepsilon \nonumber \\
&\le L_{p-1}(v_t) + w(u_1,v_t) - w(u_1,v_1)   + \sum_{i=1}^{t-1} w(u_{i+1},v_i) - \sum_{i=2}^{t-1} w(u_i,v_i) + (t-1)\varepsilon \nonumber \\
\implies &\sum_{i=1}^t w(u_{i},v_{i})  \le
 \sum_{i=2}^t w(u_{i},v_{i-1}) + w(u_1,v_t) + (t-1)\varepsilon  \nonumber\\
\end{align}
The proof is completed by adding \eqref{eq:augPath} and \eqref{eq:augCycle} for all vertex disjoint augmenting paths and cycles.
\end{proof}

\bibliographystyle{plain}
\bibliography{references}
\end{document}